%% file: paper.tex
\documentclass[conference]{IEEEtran}
\IEEEoverridecommandlockouts
\usepackage{cite}
\usepackage{amsmath,amssymb,amsfonts}
\usepackage{algorithmic}
\usepackage{graphicx}
\usepackage{textcomp}
\usepackage{xcolor}
\input{header.tex}

\newcommand{\changedd}[1]{\textcolor{black}{#1}}

\def\BibTeX{{\rm B\kern-.05em{\sc i\kern-.025em b}\kern-.08em
    T\kern-.1667em\lower.7ex\hbox{E}\kern-.125emX}}

\IEEEoverridecommandlockouts\IEEEpubid{\makebox[\columnwidth]{ 978-1- 979-8-3503-1090-0/23/\$31.00~\copyright~2023 European Union\hfill} \hspace{\columnsep}\makebox[\columnwidth]{ }}

\makeatletter
\def\ps@IEEEtitlepagestyle{%
  \def\@oddfoot{\mycopyrightnotice}%
  \def\@oddhead{\hbox{}\@IEEEheaderstyle\leftmark\hfil\thepage}\relax
  \def\@evenhead{\@IEEEheaderstyle\thepage\hfil\leftmark\hbox{}}\relax
  \def\@evenfoot{}%
}
\def\mycopyrightnotice{%
  \begin{minipage}{\textwidth}
  \centering \scriptsize
  Copyright~\copyright~20XX IEEE.  Personal use of this material is permitted.  Permission from IEEE must be obtained for all other uses, in any current or future media, including reprinting/republishing this material for advertising or promotional purposes, creating new collective works, for resale or redistribution to servers or lists, or reuse of any copyrighted component of this work in other works.
  \end{minipage}
}
\makeatother

\begin{document}

\title{Optimal Linear Precoder Design for MIMO-OFDM Integrated Sensing and Communications Based on Bayesian Cram\'er-Rao Bound
\thanks{
    The authors were supported in part by the
    German Federal Ministry of Education and Research (BMBF)
    in the programme “Souver\"an. Digital. Vernetzt.”
    within the research hub 6G-life under Grant 16KISK002.
    U. M\"onich and H. Boche were supported in part by the BMBF within the project "Post Shannon Communication - NewCom" under Grant 16KIS1003K.
    }
}

\author{\IEEEauthorblockN{Xinyang Li\IEEEauthorrefmark{1}\IEEEauthorrefmark{2},
Vlad C. Andrei\IEEEauthorrefmark{1}\IEEEauthorrefmark{3}, Ullrich J. M\"onich\IEEEauthorrefmark{1}\IEEEauthorrefmark{4} and
Holger Boche\IEEEauthorrefmark{1}\IEEEauthorrefmark{5}}
\IEEEauthorblockA{\IEEEauthorrefmark{1}Chair of Theoretical Information Technology, Technical University of Munich, Munich, Germany\\
\IEEEauthorrefmark{1}BMBF Research Hub 6G-life,
\IEEEauthorrefmark{5}Munich Center for Quantum Science and Technology,
\IEEEauthorrefmark{5}Munich Quantum Valley\\
Email: \IEEEauthorrefmark{2}xinyang.li@tum.de,
\IEEEauthorrefmark{3}vlad.andrei@tum.de,
\IEEEauthorrefmark{4}moenich@tum.de,
\IEEEauthorrefmark{5}boche@tum.de}
}

\maketitle

\begin{abstract}
In this paper, we investigate the fundamental limits of MIMO-OFDM \ac{isac} systems based on a \ac{bcrb} analysis. We derive the \ac{bcrb} for joint channel parameter estimation and data symbol detection, in which a performance trade-off between both functionalities is observed. We formulate the optimization problem for a linear precoder design and propose the \ac{srgd} approach to solve the non-convex problem. We analyze the optimality conditions and show that \ac{srgd} ensures convergence with high probability. The simulation results verify our analyses and also demonstrate a fast convergence speed. Finally, the performance trade-off is illustrated and investigated.
\end{abstract}


\glsresetall

\section{Introduction}
\input{paper-sections/intro.tex}

\section{System Model}
\input{paper-sections/sysmodel.tex}

\section{Problem Formulation}
\input{paper-sections/probform}

\section{Optimization}
\input{paper-sections/opt.tex}

\section{Numerical Results}
\input{paper-sections/numres}

\section{Conclusion}
\input{paper-sections/conclusion}

\bibliographystyle{IEEEtran}
\bibliography{IEEEabrv,mybib}

\end{document}

%% file: header.tex
\usepackage{amsmath, bm, amssymb, algorithm, mathtools, svg,ifluatex, pgfplots, amsthm, caption, subcaption,cite,multirow,url}
\usepackage[acronym,shortcuts]{glossaries}
\allowdisplaybreaks
\usepackage[bookmarks=false]{hyperref}
\pgfplotsset{compat=1.15}

 \newcommand{\Hquad}{\hspace{0.5em}} 

\newcommand{\br}[1]{\left( #1 \right)}
\newcommand{\brsq}[1]{\left[ #1 \right]}
\newcommand{\brcur}[1]{\left\{ #1 \right\}}

\newcommand{\ar}{\bm{a}_{R}}
\newcommand{\at}{\bm{a}_{T}}
\newcommand{\RR}{\mathbb{R}}
\newcommand{\CC}{\mathbb{C}}

\newcommand{\herm}{^\mathrm{H}}

\newcommand{\expc}[2]{\mathbb{E}_{#1}\brsq{#2}}
\newcommand{\tr}[1]{\mathrm{tr}\br{#1}}
\newcommand{\pard}[2]{\frac{\partial {#1}}{\partial {#2}}}
\newcommand{\re}[1]{\mathrm{Re}\br{#1}}
\newcommand{\im}[1]{\mathrm{Im}\br{#1}}
\newcommand{\proj}[2]{\mathcal{P}_{#1}\br{#2}}

\newcommand{\fim}{\bm{\mathcal{I}}}

\newcommand{\bH}{\bm{H}}
\newcommand{\bV}{\bm{V}}
\newcommand{\bA}{\bm{A}}
\newcommand{\bS}{\bm{S}}

\newcommand{\bJ}{\bm{J}}
\newcommand{\bx}{\bm{x}}
\newcommand{\bs}{\bm{s}}

\newcommand{\by}{\bm{y}}
\newcommand{\bY}{\bm{Y}}
\newcommand{\bz}{\bm{z}}

\newcommand{\bC}{\bm{C}}
\newcommand{\bW}{\bm{W}}
\newcommand{\bT}{\bm{T}}
\newcommand{\bR}{\bm{R}}
\newcommand{\bLam}{\bm{\Lambda}}
\newcommand{\bzero}{\bm{0}}
\newcommand{\eye}{\bm{I}}
\newcommand{\nt}{N_t}
\newcommand{\nr}{N_r}
\newcommand{\ns}{N_s}
\newcommand{\bb}{\bm{b}}
\newcommand{\btau}{\bm{\tau}}
\renewcommand{\bf}{\bm{f}}
\newcommand{\bthe}{\bm{\theta}}
\newcommand{\bphi}{\bm{\phi}}
\newcommand{\bxi}{\bm{\xi}}
\newcommand{\bzeta}{\bm{\zeta}}
\newcommand{\bmu}{\bm{\mu}}

\newcommand{\calS}{\mathcal{S}}
\newcommand{\hf}{\hat{f}}
\newcommand{\cL}{\mathcal{L}}

\newcommand{\opt}{\mathrm{opt}}
\newcommand{\grad}{\mathrm{grad}}

\newacronym{mimo}{MIMO}{multiple-input multiple-output}
\newacronym{ofdm}{OFDM}{orthogonal frequency-division multiplexing}
\newacronym{5g}{5G}{fifth generation}
\newacronym{6g}{6G}{sixth generation}
\newacronym{b5g}{B5G}{Beyond 5G}
\newcommand{\mimoofdm}{\ac{mimo}-\ac{ofdm} }
\newacronym{crlb}{CRB}{Cram\'er-Rao bound}
\newacronym{bcrb}{BCRB}{Bayesian Cram\'er-Rao bound}
\newacronym{bfim}{BFIM}{Bayesian Fisher information matrix}
\newacronym{repms}{REPMS}{Riemannian Exact Penalty Method via Smoothing}
\newacronym{srepms}{SREPMS}{Stochastic Riemannian Exact Penalty Method via Smoothing}
\newacronym{sdr}{SDR}{semidefinite relaxation}
\newacronym{sdp}{SDP}{semidefinite programming}
\newacronym{mmwave}{mmWave}{millimeter wave}
\newacronym{isac}{ISAC}{integrated sensing and communications}
\newacronym{aoa}{AoA}{angle of arrival}
\newacronym{aod}{AoD}{angle of departure}
\newacronym{qos}{QoS}{Quality of Service}
\newacronym{th}{THz}{Terahertz}
\newacronym{dof}{DoFs}{degrees of freedom}
\newacronym{iot}{IoT}{Internet of Things}
\newacronym{uav}{UAV}{unmanned aerial vehicle}
\newacronym{tx}{Tx}{transmitter}
\newacronym{rx}{Rx}{receiver}
\newacronym{ula}{ULA}{uniform linear array}
\newacronym{dft}{DFT}{Discrete Fourier Transform}
\newacronym{re}{RE}{resource element}
\newacronym{dl}{DL}{Downlink}
\newacronym{ul}{UL}{Uplink}
\newacronym{psd}{PSD}{positive semidefinite}
\newacronym{mse}{MSE}{mean squared error}
\newacronym{svd}{SVD}{singular value decomposition}
\newacronym{cpu}{CPU}{central processing unit}
\newacronym{mui}{MUI}{multi-user interference}
\newacronym{mi}{MI}{mutual information}
\newacronym{srmo}{SRMO}{stochastic Riemannian manifold optimization}
\newacronym{srgd}{SRGD}{stochastic Riemannian gradient descent}
\newacronym{srcg}{SRCG}{stochastic Riemannian conjugate gradient}
\newacronym{kkt}{KKT}{Karush–Kuhn–Tucker}
\newacronym{as}{a.s.}{almost sure}
\newacronym{papr}{PAPR}{peak to average power ratio}

\newtheorem{theorem}{Theorem}

\newtheorem{corollary}{Corollary}

\glsdisablehyper
\allowdisplaybreaks

%% file: paper-sections/intro.tex
\changedd{With the increasing demand for massive device connection and high-speed communications,} current wireless networks face many challenges, such as interference and spectrum sparsity. To overcome these problems, \ac{isac} is envisioned as a key technology in next-generation wireless systems\cite{isacsurvy}. Due to similar platform requirements and signal processing algorithms\cite{Overview}, it becomes more and more desirable to merge both functionalities to further explore performance potential and reduce hardware cost. Additionally, \ac{isac} facilitates coordination between both functionalities through mutual assistance to improve individual performance such as accuracy, robustness, and security\cite{isacsecurity}.

\Ac{mimo} and \ac{ofdm} are both powerful techniques that enable the efficient utilization of spectral and spatial resources. They have been standardized in most wireless systems such as Wi-Fi and 5G, and are still expected to be important signaling techniques in the next-generation systems, e.g., Beyond 5G and 6G\cite{holger}. Therefore, implementing \ac{isac} in \mimoofdm systems is of great interest and is able to be adopted in most applications, e.g., \ac{iot}\cite{iot} and perceptive mobile networks\cite{perceptive}.

Waveform design plays a crucial role in \ac{isac} systems as it directly impacts the theoretical system performance. Designing waveforms with favorable properties usually involves solving optimization problems that are formulated based on certain criteria. For example, the authors in \cite{papr} formulate the design problem to minimize the multi-user interference while maintaining the radar beampattern and \ac{papr} constraints, and in\cite{optimized, Li2305:Optimal} the sensing symbols are designed and allocated to an \ac{ofdm} grid to lower the \ac{crlb} of sensing parameters. However, \changedd{the fundamental limits for \ac{isac} are still not well studied\cite{fundamental},} due to the different performance requirements of both functionalities. Moreover, the formulated problems are usually non-convex and are therefore difficult to solve. In the cases where suboptimal solutions are targeted, convergence guarantees and the speed of the adopted algorithms are both critical factors to be considered.

In this work, we investigate the \mimoofdm \ac{isac} performance from the perspective of \ac{bcrb}, which generalizes the \ac{crlb} by considering the prior distribution of parameters being estimated. In particular, we consider \ac{isac} to be a task of joint channel parameter estimation and data symbol detection and derive its \ac{bfim}, the inverse of the \ac{bcrb}. We point out the performance trade-off between both functionalities by taking the log determinant of the \ac{bfim}. Based on that, the optimization problem to design the linear precoder is formulated. We propose solving the non-convex stochastic problem using the \ac{srgd} method and prove that its convergence ensures the optimality conditions. Simulation results\footnote{Code is available at \url{https://github.com/xinyanglii/isac-mimo-ofdm-wf}.} highlight the advantages of \ac{srgd} and validate our analyses.

%% file: paper-sections/sysmodel.tex
\subsection{Signal Model}

We study a \mimoofdm system where the \ac{tx} and the \ac{rx} each have a \ac{ula} with $\nt$ and $\nr$ antenna elements, respectively. The \ac{tx} allocates a set of $M$ \acp{re} $\{(n_m,k_m)\}_{m=1}^M$ for data transmission, in which we denote $(n_m,k_m)$ as the index of the $n_m$-th subcarrier and the $k_m$-th \ac{ofdm} symbol. These \acp{re} are filled with $M$ unknown data symbols $\bs_m\in\CC^{\ns}$ with $\ns$ being the number of streams. We assume that the symbols follow a complex Gaussian distribution with zero mean and an identity covariance matrix, i.e., $\bs_m \sim \mathcal{CN}(\bzero,\eye_{\ns})$ for all $m$. By applying a linear precoder $\bW\in \CC^{\nt \times \ns}$, the data symbols are mapped to the $\nt$ antennas via $\bx_m = \bW \bs_m$ for transmission, and the received signal \changedd{after going through the channel can be expressed} as\cite{Overview}
\begin{equation}
    \by_m = \bH_m\bx_m + \bz_m,
\end{equation}
where $\bH_m\in\CC^{\nr\times\nt}$ is the channel matrix on the \mbox{$m$-th} \ac{re} and $\bz_m \in \CC^{\nr}$ is complex white Gaussian noise following $\mathcal{CN}(\bm{0}, \sigma_z^2\eye_{\nr})$ for all $m$. By considering the sparse propagation paths in the high frequency range\cite{mmwave}, the channel matrix characterized by multipath parameters takes the form
\begin{equation}
    \bH_m(\bxi) = \sum_{l=1}^Lb_l\omega_{m,l}\ar(\phi_l)\at(\theta_l)^\top, \label{eq:channel}
\end{equation}
where $\omega_{m,l} = e^{-j2\pi n_mf_0 \tau_l}e^{j2\pi f_{D,l}k_mT_s}$,  $f_0$ is the \ac{ofdm} subcarrier spacing, $T_s$ is the \ac{ofdm} symbol duration, $L$ is the number of propagation paths, $b_l$, $\tau_l$, $f_{D,l}$, $\theta_l$, $\phi_l$ are the channel gain, path delay, Doppler shift, \ac{aod} and \ac{aoa} of the $l$-th path, $\at(\cdot) \in \CC^{N_T}$ and $\ar(\cdot)\in \CC^{N_R}$ are the transmit and receive \ac{ula} response vector, respectively. $\bxi$ contains all the multipath parameters and is defined as $\bxi = [\bb_R^\top, \bb_I^\top, \btau^\top, \bf_D^\top, \bthe^\top, \bphi^\top]^\top \in \RR^{6L}$, with each component grouping the respective parameters of all $L$ paths. To make $\bxi$ a real vector, we split the path gains into real and imaginary parts and stack them into $\bb_R$ and $\bb_I$. 

As one of the scenarios in \ac{isac}, the \ac{rx} aims to estimate channel parameters and detect data symbols jointly. Toward this goal, we utilize the \ac{bcrb}\cite{bcrb, van2007bayesian}, which states a lower bound for the \ac{mse} for any estimator. 

\subsection{Bayesian Fisher Information Matrix}
We stack the channel parameters and data symbols into a single vector as $\bzeta = [\bxi^\top, \tilde{\bs}_1^\top, \cdots, \tilde{\bs}_M^\top]^\top \in \RR^{6L + 2N_sM}$, where $\tilde{\bs}_m = [\bs_{m,R}^\top, \bs_{m,I}^\top]^\top$, $\bs_m = \bs_{m,R} + j\bs_{m,I}$. The \ac{isac} system can then be considered as a system with input $\bzeta$, and output $\bY = [\by_1, \by_2, ... ,\by_M]$ with the joint distribution
\begin{equation}
\begin{split}
    p(\bY,\bzeta) &= p(\bY|\bzeta)p(\bzeta) = p(\bzeta)\prod_{m=1}^{M} p(\by_m|\bzeta).
\end{split}\label{eq:dist}
\end{equation} 
For any estimator $\hat{\bzeta}(\bY)$ of $\bzeta$ and the resulting \ac{mse}, we have the \ac{bcrb} as its lower bound, i.e., 
\begin{equation}
\label{eq:msebound}
    \expc{}{\br{\bzeta-\hat{\bzeta}(\bY)}\br{\bzeta-\hat{\bzeta}(\bY)}^\top} \succeq \fim^{-1}
\end{equation}
where $\bA\succeq\bm{B}$ indicates $\bA-\bm{B}$ is a \ac{psd} matrix, $\fim$ is the \ac{bfim} and is defined as\cite{mutual}
\begin{equation}
    \fim = \expc{\bY, \bzeta}{\pard{\log p(\bY,\bzeta)}{\bzeta} \br{\pard{\log p(\bY,\bzeta)}{\bzeta}}^\top}.
\end{equation}
By making use of \eqref{eq:dist}, the \ac{bfim} can be decomposed into two parts, i.e.,
\begin{equation}
\begin{split}
    \fim &= \expc{\bzeta}{\fim_c} + \fim_p,
\end{split}
\end{equation}
with $\fim_c$ and $\fim_p$ denoting the conditional and prior FIM respectively and taking the forms
\begin{equation}
\begin{split}
    \fim_c &= \expc{\bY|\bzeta}{\pard{\log p(\bY| \bzeta)}{\bzeta} \br{\pard{\log p(\bY| \bzeta)}{\bzeta}}^\top},\\
    \fim_p &= \expc{\bzeta}{\pard{\log p(\bzeta)}{\bzeta} \br{\pard{\log p(\bzeta)}{\bzeta}}^\top}.
\end{split}
\end{equation}

We assume a Gaussian prior distribution $p(\bzeta)=\mathcal{N} (\bmu_{\bzeta}, \bC_{\bzeta})$, where $\bmu_{\bzeta}$ and $\bC_{\bzeta}$ can be chosen practically. For instance, when the sequential estimation and tracking are applied, the prior mean of channel parameters $\bmu_{\bxi}$ can be set to the previous state, and their prior covariance matrix $\bC_{\bxi}$ contains the uncertainty in them. Furthermore, we have $\tilde{\bs}_m$ following $\mathcal{N}(\bzero, \frac{1}{2}\eye_{2\ns})$, and assume $\tilde{\bs}_m$ and $\bxi$ are independent, such that $\bC_{\bzeta}$ is a block diagonal matrix with $\bC_{\bxi}$ and $\frac{1}{2}\eye_{2\ns}$ on the diagonal. The conditional distribution $p(\by_m|\bzeta)$ is a complex Gaussian distribution with mean $\bmu_m = \bH_m \bx_m$ and covariance matrix $\sigma_z^2\eye_{\nr}$. Hence, we have\cite{kay1993fundamentals}
\begin{equation}
    \fim_c= \frac{2}{\sigma_z^2}\sum_{m=1}^M \re{\pard{\bmu_m\herm}{\bzeta}\pard{\bmu_m}{\bzeta}}.  \label{eq:covexpc}
\end{equation}
It can also be shown that\cite{fundamental}
\begin{equation}
\begin{split}
    \fim_c = \begin{bmatrix}
        \fim^s_c(\bxi, \bxi) & \fim^s_c(\bxi, \tilde{\bs}_1) & \cdots &  \fim^s_c(\bxi, \tilde{\bs}_M) \\
        \fim^s_c(\tilde{\bs}_1, \bxi) & \fim^s_c(\tilde{\bs}_1, \tilde{\bs}_1) &  \cdots & \fim^s_c(\tilde{\bs}_1, \tilde{\bs}_M) \\
        \vdots & \vdots & \ddots & \vdots\\
        \fim^s_c(\tilde{\bs}_M, \bxi) & \fim^s_c(\tilde{\bs}_M, \tilde{\bs}_1) & \cdots &  \fim^s_c(\tilde{\bs}_M, \tilde{\bs}_M)
    \end{bmatrix}
\end{split} \label{eq:fimcomp}
\end{equation}
with the arguments inside the bracket of the submatrix $\fim^s_c(\cdot,\cdot)$ denoting to which components of $\bzeta$ the derivatives in \eqref{eq:covexpc} are taken. In \cite{Li2305:Optimal} it is derived that
\begin{align}
    \pard{\bm{\mu}_m}{\bm{\xi}} &= (\bm{R} \bm{\Lambda}_{m}) * (\bx_m^\top\bm{T})\\
    \fim^s_c(\bxi, \bxi) &= \frac{2}{\sigma_z^2}\sum_{m=1}^M\re{ \bm{\Lambda}_m\herm \bT\herm \bx_m^*\bx_m^\top \bT\bLam_m\circ \br{\bR\herm \bR}},
\end{align}
where we define $*$ and $\circ$ as the Khatri-Rao product and the Hadamard product, respectively, and the related matrices are given in \eqref{eq:fimparam}. In addition,
\begin{equation}
    \pard{\bm{\mu}_m}{\tilde{\bm{s}}_{m'}} = \begin{cases}
    \begin{bmatrix}
        \bH_m & \bH_m
    \end{bmatrix}\begin{bmatrix}
        \bm{W} & \\
        & j\bm{W}
    \end{bmatrix}= \tilde{\bm{H}}_m\tilde{\bm{W}}, & m=m'\\
    \bzero, & m\neq m'
    \end{cases}
\end{equation}
\begin{equation}
    \fim^s_c(\tilde{\bm{s}}_m, \tilde{\bm{s}}_m) = \frac{2}{\sigma_z^2}\re{\tilde{\bm{W}}\herm\tilde{\bm{H}}_m\herm\tilde{\bm{H}}_m\tilde{\bm{W}}}
\end{equation}
\begin{figure*}[ht]
    \begin{equation}
    \begin{split}
        \bm{\Lambda}_m &= \text{blkdiag}\brcur{\bm{\Omega}_m, j\bm{\Omega}_m, \bm{G}_m\bm{B}, \bm{F}_m\bm{B}, \bm{\Omega}_m\bm{B}, \bm{\Omega}_m\bm{B}},\quad \bm{B} = \text{diag}\brcur{b_1\cdots b_L}, \quad \bm{\Omega}_m = \text{diag}\brcur{\omega_{m,1} \cdots \omega_{m,L}}\\
        \bm{G}_m &= \text{diag}\brcur{u_{m,1}\cdots u_{m,L}},\Hquad \bm{F}_m = \text{diag}\brcur{v_{m,1}\cdots v_{m,L}},\Hquad  u_{m,l} = \pard{\omega_{m,l}}{\tau_l} = -j2\pi nf_0\omega_{m,l},\Hquad v_{m,l} = \pard{\omega_{m,l}}{f_{D,l}} = j2\pi kT_s\omega_{m,l},\\ 
        \bm{T} &= \brsq{\bm{A}_T, \bm{A}_T, \bm{A}_T, \bm{A}_T, \bm{A}_T, \bm{D}_T},\quad \bm{A}_T = \brsq{\bm{a}_T(\theta_1) \cdots \bm{a}_T(\theta_L)},\quad \bm{D}_T = \brsq{\bm{d}_T(\theta_1) \cdots \bm{d}_T(\theta_L)}, \quad \bm{d}_T(\theta_l) = \pard{\at(\theta_l)}{\theta_l},\\
        \bm{R} &= \brsq{\bm{A}_R, \bm{A}_R, \bm{A}_R, \bm{A}_R, \bm{D}_R, \bm{A}_R}, \quad \bm{A}_R = \brsq{\bm{a}_R(\phi_1) \cdots \bm{a}_R(\phi_L)}, \quad \bm{D}_R = \brsq{\bm{d}_R(\phi_1) \cdots \bm{d}_R(\phi_L)}, \quad \bm{d}_R(\phi_l) = \pard{\ar(\phi_l)}{\phi_l}. 
    \end{split}
    \label{eq:fimparam}
\end{equation}
\hrulefill
\end{figure*}

Therefore, taking the expectation with respect to each component in \eqref{eq:fimcomp} leads to
\begin{equation}
\begin{split}
    &\expc{\bm{\zeta}}{\fim^s_c(\bm{\xi}, \bm{\xi})}\\ 
    &\Hquad= \frac{2}{\sigma_z^2}\sum_{m=1}^M\expc{\bm{\xi}}{\re{ \bm{\Lambda}_m\herm \bm{T}\herm \expc{\bs_m}{\bm{x}_m^*\bm{x}_m^\top} \bm{T\Lambda}_m}\circ \br{\bm{R}\herm \bm{R}}}\\
    &\Hquad= \frac{2}{\sigma_z^2}\sum_{m=1}^M\mathbb{E}_{\bm{\xi}}\biggl[\underbrace{\re{ \bm{\Lambda}_m\herm \bm{T}\herm \bm{W}^*\bm{W}^\top \bm{T\Lambda}_m}\circ \br{\bm{R}\herm \bm{R}}}_{\bm{S}_{m}(\bm{W},\bm{\xi})}\biggr],
\end{split}
\label{eq:Sm}
\end{equation}
\begin{equation}
\begin{split}
    &\expc{\bm{\zeta}}{\fim^s_c(\tilde{\bm{s}}_m, \tilde{\bm{s}}_{m'})}\\
    &\Hquad= \begin{cases}
        \frac{2}{\sigma_z^2}\re{\tilde{\bm{W}}\herm\expc{\bm{\xi}}{\tilde{\bm{H}}_m\herm\tilde{\bm{H}}_m}\tilde{\bm{W}}}, & m=m'\\
        \bm{0},& m\neq m',
    \end{cases} \\
\end{split}\label{eq:fimc}
\end{equation}
\begin{equation}
\begin{split}
    &\expc{\bm{\zeta}}{\fim^s_c(\bm{\xi}, \tilde{\bm{s}}_m)}\\
    &\Hquad = \frac{2}{\sigma_z^2}\sum_{m=1}^M\expc{\bm{\zeta}}{\re{(\bm{R} \bm{\Lambda}_{m}) * (\bx_m^\top\bm{T}) \tilde{\bm{H}}_m\tilde{\bm{W}}}}\\
    & \Hquad= \bm{0},
\end{split}
\end{equation}
where the last step is due to the mean of $\bx_m$ being zero. The second term $\fim_p$ is straightforward to be obtained as
\begin{equation}
    \fim_p = \bm{C}_{\bm{\zeta}}^{-1} = \begin{bmatrix}
        \bm{C}_{\bm{\xi}}^{-1} & \bm{0} \\
        \bm{0} & 2\bm{I}_{2N_sM}
    \end{bmatrix}.
\end{equation}
Consequently, the \ac{bfim} is a block diagonal matrix with nonzero matrices $\expc{\bm{\zeta}}{\fim^s_c(\bm{\xi}, \bm{\xi})} + \bm{C}_{\bm{\xi}}^{-1}$ and $\expc{\bm{\zeta}}{\fim^s_c(\tilde{\bm{s}}_m, \tilde{\bm{s}}_{m})} + 2\bm{I}_{2N_s}$ for all $m$ on the diagonal.

%% file: paper-sections/probform.tex
The \ac{bfim} can be optimized using different criteria\cite{maxdet}. In this work, we choose the log determinant as the objective function, which results in 
\begin{equation}
\begin{split}
    \log\det\br{\fim} &= \log\det\br{\expc{\bm{\zeta}}{\fim^s_c(\bm{\xi}, \bm{\xi})} + \bm{C}_{\bm{\xi}}^{-1}}\\
    &+ \sum_{m=1}^M\log\det\br{\expc{\bm{\zeta}}{\fim^s_c(\tilde{\bm{s}}_m, \tilde{\bm{s}}_{m})} + 2\bm{I}_{2N_s}}.
    \label{eq:obj1}
\end{split}
\end{equation}
By plugging the expression of $\expc{\bm{\zeta}}{\fim^s_c(\tilde{\bm{s}}_m, \tilde{\bm{s}}_{m})}$ from~\eqref{eq:fimc} into the second term, we have
\begin{equation}
\begin{split}
    &\log\det\br{\expc{\bm{\zeta}}{\fim^s_c(\tilde{\bm{s}}_m, \tilde{\bm{s}}_{m})} + 2\bm{I}_{2N_s}} \\
    &\quad=\log\det\br{\frac{2}{\sigma_z^2}\re{\tilde{\bm{W}}\herm\expc{\bm{\xi}}{\tilde{\bm{H}}_m\herm\tilde{\bm{H}}_m}\tilde{\bm{W}}}+ 2\bm{I}_{2N_s}}\\
    &\quad =\log\det \br{\re{\begin{bmatrix}
        \bA + 2\bm{I}_{N_s}& j\bA\\
        -j\bA & \bA + 2\bm{I}_{N_s}
    \end{bmatrix}}}\\
    & \quad =\log\det \br{\begin{bmatrix}
        \re{\bA + 2\bm{I}_{N_s}} & -\im{\bA+ 2\bm{I}_{N_s}}\\
        \im{\bA+ 2\bm{I}_{N_s}} & \re{\bA + 2\bm{I}_{N_s}}
    \end{bmatrix}}\\
    &\quad=2\log\det \br{\bA+ 2\bm{I}_{N_s}}\\
    &\quad=2\log\det \br{\frac{2}{\sigma_z^2}\bm{W}\herm\expc{\bm{\xi}}{\bH_m\herm\bH_m}\bm{W}+ 2\bm{I}_{N_s}}, \label{eq:mi}
\end{split}
\end{equation}
with $\bA = \frac{2}{\sigma_z^2}\bm{W}\herm\expc{\bm{\xi}}{\bH_m\herm\bH_m}\bm{W}$. It's not surprising that \eqref{eq:mi} has a similar form with the mutual information between $\bs_m$ and $\by_m$. In fact, if we omit the factor $2$ inside and outside the $\log\det$, \eqref{eq:mi} is an upper bound on the ergodic mutual information between $\bs_m$ with $\by_m$, since
\begin{equation}
\begin{split}
    \expc{\bm{\xi}}{I(\bs_m, \by_m)} &= \expc{\bm{\xi}}{\log\det \br{\frac{1}{\sigma_z^2}\bm{W}\herm\bH_m\herm\bH_m\bm{W} + \bm{I}_{N_s}}}\\
    &\le \log\det \br{\frac{1}{\sigma_z^2}\bm{W}\herm\expc{\bm{\xi}}{\bH_m\herm\bH_m}\bm{W}+ \bm{I}_{N_s}}
\end{split}
\end{equation}
using the Jensen inequality, and $I(\cdot,\cdot)$ denotes the mutual information. The first term in~\eqref{eq:obj1} is nothing else than the $\log\det$ of the \ac{bfim} for channel parameter estimation. Therefore, we observe a performance trade-off between both functionalities in which the communication quality is measured in terms of the upper bound of the ergodic mutual information, channel sensing is assessed by the parameter estimation error bound, and both are bridged through the \ac{bcrb} of joint parameter estimation and symbol detection. By introducing an adjustable weighting factor to~\eqref{eq:obj1}, we are able to take over the control of such a trade-off while optimizing.

We also take into consideration different importance levels of different channel parameters in practice. To do this, we multiply the \ac{bfim} for the parameter part by a weighting matrix $\bJ$\cite{maxdet}. Therefore, the resulting optimization problem based on~\eqref{eq:obj1} is formulated as
\begin{equation}
\begin{split}
    &\max_{\bW}\Hquad f(\bm{W};\alpha) =\max_{\bW}\Hquad \alpha f_s(\bW) + (1-\alpha) \frac{1}{M} f_c(\bW) \\
    =& \max_{\bm{W}} \alpha\log\det \br{ \bm{J}\herm\br{\frac{2}{\sigma_z^2}\sum_{m=1}^M \expc{\bm{\xi}}{\bS_m(\bW,\bxi)} + \bm{C}_{\bm{\xi}}^{-1}}\bm{J}}\\
    &\Hquad +  \frac{1 - \alpha}{M}\sum_{m=1}^M 2\log\det \br{\frac{2}{\sigma_z^2}\bm{W}\herm\expc{\bm{\xi}}{\bH_m\herm\bH_m}\bm{W}+ 2\bm{I}}\\
    &\Hquad \text{s. t. }\quad \tr{\bm{WW}\herm} = \|\bm{W}\|_F^2 \le P,
\end{split}
\tag{P1} \label{P1}
\end{equation}
with $\alpha\in [0,1]$ the trade-off factor, $\bS_m(\bW,\bxi)$ given in \eqref{eq:Sm}, and $f_s(\bW)$, $f_c(\bW)$ indicating performance metrics for sensing and communications, respectively. To scale both quantities to the same order, $f_c(\bW)$ is multiplied by $\frac{1}{M}$. Such a problem is difficult to solve for two reasons. The first one is that the expectation in the objective function is intractable to derive. The stochastic optimization methods\cite{hannah2015stochastic} can be applied in which the expectation is approximated by sampling $\bxi$ from the prior distribution, as will be discussed later. The other reason is that it is a non-convex problem if $\bW$ is a tall matrix, i.e., $N_s < N_t$. This issue can be addressed partly by \ac{sdr}\cite{sdr}, but as $N_t$ and the number of sampling points become larger, \ac{sdr} suffers from higher computational costs and becomes infeasible to solve in practice\cite{Li2305:Optimal}. To this end, we observe that $f(\bW;\alpha)$ is maximized only if the transmit power $P$ is exhausted due to the monotonicity of $f$ in the scaling of $\bW$, which leads to the complex hypersphere manifold
\begin{equation}
    \calS=\brcur{\bm{W}\in\CC^{N_t\times N_s} \middle| \tr{\bm{WW}\herm}=P}.
    \label{eq:sphere}
\end{equation}
Hence, \eqref{P1} can be solved using Riemannian manifold optimization methods\cite{absil2009optimization} on $\calS$. Taking into account the stochastic approximation and defining
\begin{equation}
\begin{split}
    \bm{S}_{m}^N(\bm{W}) &= \frac{1}{N}\sum_{n=1}^N\bm{S}_{m}(\bm{W},\bm{\xi}_n) \approx \expc{\bm{\xi}}{\bm{S}_{m}(\bm{W},\bm{\xi})}\\
    \bm{K}_m^N &= \frac{1}{N}\sum_{n=1}^N \bH_m(\bm{\xi}_n)\herm\bH_m(\bm{\xi}_n) \approx \expc{\bm{\xi}}{\bH_m\herm\bH_m},
\end{split}
\end{equation}
with $N$ the number of sampling points, and $\bxi_n$ the $n$-th sampled channel parameters, we end up with the final \ac{srmo} problem
\begin{equation}
\begin{split}
    &\max_{\bW\in\calS}\Hquad \hf(\bm{W};\alpha) =\max_{\bW}\Hquad \alpha\hf_s(\bW) + (1-\alpha)\frac{1}{M}\hf_c(\bW)\\
    =& \max_{\bW\in\calS} \alpha\log\det \br{ \bm{J}\herm\br{\frac{2}{\sigma_z^2}\sum_{m=1}^M \bm{S}_{m}^N(\bm{W}) + \bm{C}_{\bm{\xi}}^{-1}}\bm{J}} \\
    &\Hquad+  \frac{1-\alpha}{M}\sum_{m=1}^M 2\log\det \br{\frac{2}{\sigma_z^2}\bm{W}\herm\bm{K}_m^N\bm{W}+ 2\bm{I}}
\end{split}\tag{PM} \label{PM}
\end{equation}
where we use $\hat{(\cdot)}$ to indicate the stochastic version of a function. \eqref{PM} can be solved efficiently by applying \ac{srgd}. In the following, we analyze the optimality conditions for the formulated problems and show that \ac{srgd} ensures the convergence to a critical point of \eqref{P1}.

%% file: paper-sections/opt.tex
To state the optimality conditions for the constrained problems above, we define the Lagrangian function as
\begin{equation}
\label{eq:lag}
    \cL(\bW,\lambda) = -f(\bW;\alpha) + \lambda \br{\tr{\bW\bW\herm}-P},
\end{equation}
with $\lambda$ being the Lagrangian multiplier. Let $\nabla$ denote the complex gradient\cite{petersen2008matrix} in the Euclidean space, if we replace the power constraint in \eqref{P1} with equality, we have the following first-order optimality conditions\cite{wright1999numerical}:
\begin{theorem}[First-Order Necessary Conditions]
\label{fonc}
If $\bW_\opt$ is a local optimum of \eqref{P1}, then there exists a Lagrangian multiplier $\lambda_\opt$ such that
\begin{align}
    \nabla \cL(\bW_\opt,\lambda_\opt) &= \bzero, \label{eq:lagder}\\
    \tr{\bW_\opt\bW_\opt\herm}-P &= 0.
\end{align}
\end{theorem}
Theorem \ref{fonc} is also known as the \ac{kkt} conditions. By plugging \eqref{eq:lag} into \eqref{eq:lagder}, we have
\begin{equation}
    \nabla \cL(\bW_\opt,\lambda_\opt) = -\nabla f(\bW_\opt;\alpha) + 2\lambda_\opt\bW_\opt = \bzero.
    \label{eq:kkt}
\end{equation}

The principle of the Riemannian gradient method is similar to its Euclidean version. At the $t$-th iteration, the Riemannian gradient of $f$ at $\bW_t$ is given as\footnote{Note that we use $\grad$ to denote the Riemannian gradient while $\nabla$ indicates the Euclidean gradient.}
\begin{equation}
    \grad f(\bW_t; \alpha) = \proj{\bW_t}{\nabla f(\bW_t;\alpha)},
\end{equation}
where $\proj{\bW_t}{\cdot}$ is the projection operator that projects a vector in the Euclidean space onto the tangent space of $\calS$ at point $\bW_t$. The update function moves $\bW_t$ along the direction of $\grad f(\bW_t; \alpha)$ or its variants to the next point $\bW_{t+1} \in \calS$, which is performed by either exponential map or retraction, and the step size can be determined by any line search approach\cite{absil2009optimization}. In particular, for the complex hypersphere, the projector is defined as
\begin{equation}
    \proj{\bW}{\bV} = \bV - \re{\tr{\bW\herm\bV}}\bW.
\end{equation}
When $\grad f(\bW;\alpha)$ converges to $\bzero$, we have
\begin{equation}
    \nabla f(\bW;\alpha) - \re{\tr{\bW\herm\nabla f(\bW;\alpha)}}\bW = \bzero,
    \label{eq:gradf}
\end{equation}
and notice that it satisfies the \ac{kkt} conditions. Hence, we have the following corollary:
\begin{corollary}
If the Riemannian gradient of $f$ at $\bW_\opt$ vanishes, i.e., $\grad f(\bW_\opt; \alpha)=\bzero$, $\bW_\opt$ satisfies Theorem~\ref{fonc} with the Lagrangian multiplier being
\begin{equation}
    \lambda_\opt = \frac{1}{2}\re{\tr{\bW_\opt\herm\nabla f(\bW_\opt; \alpha)}}.
\end{equation}
\end{corollary}
\begin{proof}
By comparing \eqref{eq:kkt} and \eqref{eq:gradf} the result is obvious.
\end{proof}

Considering the \ac{srgd} for \eqref{PM}, the following theorem\cite{stochriem} guarantees its convergence with high probability:
\begin{theorem}\label{theo2}
Suppose that the \ac{srgd} applied on \eqref{PM} generates a sequence of points $\{\bW_t\}_{t>0}$ and a sequence of step sizes $\{\gamma_t\}_{t>0}$ satisfying $\sum \gamma_t^2<\infty$ and $\sum \gamma_t = +\infty$. The resulting sequence of the Riemannian gradients of $f$ converges \ac{as}, i.e., $\grad f(\bW_t;\alpha)\rightarrow \bm{0}$ \ac{as}. This holds true for either exponential map or retraction that is used as the update function.
\end{theorem} 
\begin{proof}
See\cite{stochriem}.
\end{proof}

The step size assumption in Theorem \ref{theo2} is widely used as a sufficient condition for the convergence analysis in stochastic programming\cite{eon1998online}. In simulations, we employ an adaptive line search method\cite{boumal2023intromanifolds} to determine the step size at each iteration. Although it is challenging to verify the validity of the assumption on adaptive step size, our numerical results demonstrate that the proposed method can still converge reliably and rapidly.

Based on the above analyses, we can conclude that the solution of \ac{srgd} will converge to a \ac{kkt} point of \eqref{P1} satisfying Theorem \ref{fonc} \ac{as}. In addition, the stochastic nature enables the use of a small batch size of points sampled at each iteration, thereby reducing the computational complexity.

%% file: paper-sections/numres.tex
Throughout the simulations, the system settings are chosen as $\nt = \nr = 8$, $\ns = 3$, $f_0=15$ kHz and a rectangular $128\times 14$ \ac{ofdm} grid is allocated for data transmission, i.e., $n_m=0\dots 127$, $k_m=0\dots 13$. We perform the Monte Carlo approaches, and at each simulation the channel parameters $\bxi$ are randomly generated. In particular, we set the number of propagation paths to $3$ and each parameter component is sampled according to the following rules: for all $l$, $b_l$ is sampled from the complex Gaussian distribution with zero mean and unit variance; the path distance is sampled uniformly from $10$ to $800$~meters based on which $\tau_l$ is calculated; the relative velocity is sampled uniformly from $0$ to $80$~m/s on which $f_{D,l}$ is calculated; $\phi_l$ and $\theta_l$ are both sampled uniformly from $-90^\circ$ to $90^\circ$; all components in $\bxi$ are independent of each other. The sampled parameter vector is also used as the prior mean $\bmu_{\bxi}$ and the prior covariance $\bC_{\bxi}$ is a diagonal matrix with the standard deviations for each type of parameter as follows: $10^{-1}$ for path gains, $10^{-7}$ for path delays, $50$ for Doppler shifts, $10^{-1}$ for \acp{aoa} and \acp{aod}. The scaling matrix $\bJ$ in $f_s(\bW)$ is a diagonal matrix with diagonal elements being $1/f_0$ and $f_0$ for the delay and Doppler parts respectively, and $1$ otherwise, in order to scale the values to a similar order and avoid numerical instability. 

\begin{figure}[h]
    \centering
    \includegraphics[width=.48\textwidth]{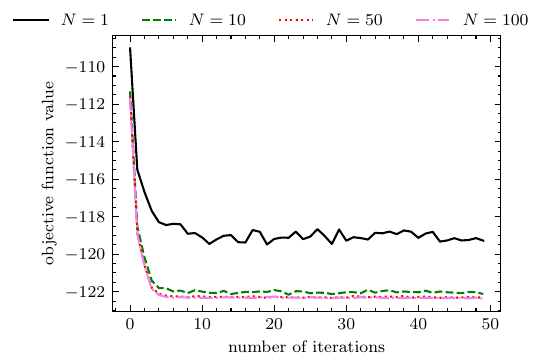}
    \caption{Convergence of $-\hf(\bW;\alpha)$ for different numbers of samples per iteration.}
    \label{fig:objcvg}
\end{figure}

We first fix the trade-off factor $\alpha$ to $0.5$ and perform the \ac{srcg} method, a variant of \ac{srgd}, on \eqref{PM} for different numbers of samples per iteration $N$. The maximum number of iterations is set to $50$ and the convergence of objective values averaged over all simulations are plotted in Fig.~\ref{fig:objcvg}. It shows that \ac{srgd} converges quickly within $10$ iterations. As $N$ becomes larger, the curve converges to a lower value, but the performance improvement is not significant as $N>10$. To validate the gradient convergence, the evolution of the Riemannian gradient norms during optimization is illustrated in Fig.~\ref{fig:graddcvg}. In addition, their standard deviations over simulations are drawn in shadow. It turns out that for a larger $N$, the Riemannian gradient can converge to $\bzero$ \ac{as}. The above results highlight the remarkable performance of \ac{srgd} in practice, especially its fast convergence and low demand for sample size per iteration.


\begin{figure}[h]
    \centering
    \includegraphics[width=.48\textwidth]{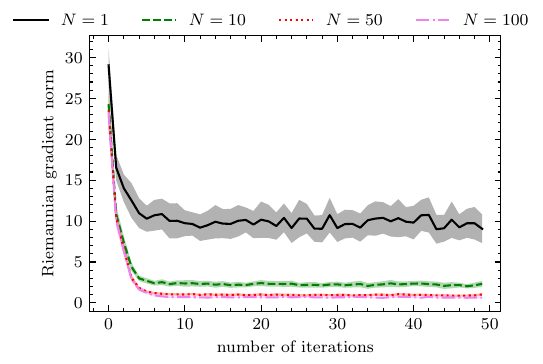}
    \caption{Riemannian gradient norm convergence for different numbers of samples per iteration. The standard deviations over all simulations are drawn in shadow.}
    \label{fig:graddcvg}
\end{figure}

Based on the optimized precoders, we calculate the values of $\hf_s(\bW)$ and $\hf_c(\bW)$ on $100$ samples and present their trade-off in Fig.~\ref{fig:trade-off} by varying the factor $\alpha$. The curves exhibit a similar behavior to the rate region for two-user \ac{mimo} systems\cite{biglieri2007mimo}, and their interior regions correspond to the achievable \ac{isac} performance. In contrast to the case of designing waveform for communications and sensing separately, in which the curve will become a straight line connecting the achievable maximum $\hf_s(\bW)$ and maximum $\hf_c(\bW)$\cite{xiong2022fundamental}, the Pareto boundary of the curves in Fig.~\ref{fig:trade-off} toward the upper right corner highlights the integration benefits gained by joint design.

\begin{figure}[h]
    \centering
    \includegraphics[width=.48\textwidth]{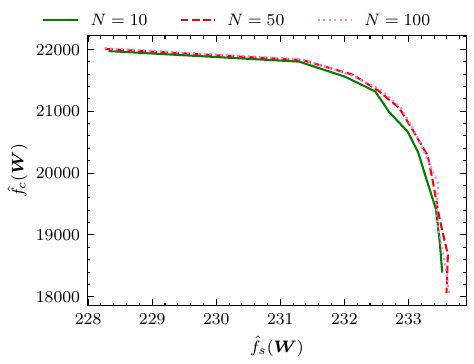}
    \caption{Trade-off between the performance of communications and sensing while varying $\alpha$ from $0$ to $1$.}
    \label{fig:trade-off}
\end{figure}

%% file: paper-sections/conclusion.tex
In this work, we address the linear precoder design problem for \mimoofdm \ac{isac} systems by optimizing the \ac{bfim} of channel parameters and data symbols. The proposed \ac{srgd} method is shown to guarantee a fast convergence with high probability. In addition, the derived objective function can be decomposed into two parts that correspond to the respective communication and sensing criteria., thereby making it possible to control the performance trade-off through a single factor. Furthermore, the performance achieved by the designed precoder is shown to be characterized by a Pareto boundary, indicating the benefits of joint system design. This observation motivates us to envision the \ac{bfim} as a potential unified metric for \ac{isac} and we leave the further investigation and study as our future work.